\documentclass[11pt]{article}
\usepackage[utf8]{inputenc}   
\usepackage[T1]{fontenc}
\usepackage{lmodern}          
\usepackage[a4paper,margin=1in]{geometry}
\usepackage{microtype}
\usepackage{setspace}         
\usepackage{titling}

\usepackage{amsmath,amssymb,amsthm,mathtools}
\usepackage{bm}
\usepackage{physics}          

\usepackage{graphicx}
\graphicspath{{figs/}}        
\usepackage{caption}
\usepackage{subcaption}
\usepackage{booktabs}
\usepackage{siunitx}
\usepackage{algorithm}
\usepackage[noend]{algpseudocode}

\usepackage[numbers,sort&compress]{natbib} 
\usepackage{hyperref}
\hypersetup{
  colorlinks=true,
  linkcolor=blue,
  citecolor=blue,
  urlcolor=blue,
  pdftitle={Algorithmic Temperature Induced by Adopted Regular Universal Turing Machine},
  pdfauthor={Kentaro Imafuku}
}
\usepackage[nameinlink,capitalize,noabbrev]{cleveref} 

\newtheorem{theorem}{Theorem}

\theoremstyle{definition}

\theoremstyle{remark}

\title{\textbf{Algorithmic Temperature Induced by Adopted Regular Universal Turing Machine}}
\author{Kentaro IMAFUKU\\
National Institute of Advanced Industrial Science and Technology (AIST)}
\date{\empty}

\begin{document}
\maketitle

\begin{abstract}
We prove that an effective temperature naturally emerges from the algorithmic structure of a regular universal Turing machine (UTM), without introducing any external physical parameter.
In particular, the redundancy growth of the machine's wrapper language induces a Boltzmann--like exponential weighting over program lengths, yielding a canonical ensemble interpretation of algorithmic probability. This establishes a formal bridge between algorithmic information theory and statistical mechanics, in which the adopted UTM determines the intrinsic ``algorithmic temperature.''
We further show that this temperature approaches its minimum limit in the binary case  under the universal mixture (Solomonoff distribution), and discuss its epistemic meaning as the resolution level of an observer.
\end{abstract}


\section{Introduction}
The parallels between computation and physics have fascinated researchers for decades~\cite{szilard1929,jaynes1957,landauer1961,bennett1982,zurek1989,leff2002}.
In statistical mechanics, temperature and equilibrium emerge from the exponential growth of microstates with energy.
A natural question is whether analogous concepts exist in the algorithmic world:
can equilibrium and temperature arise not as metaphors but as intrinsic features of computation itself?

Algorithmic information theory, initiated by Kolmogorov~\cite{kolmogorov1965} and Chaitin~\cite{chaitin1975,chaitin1987},  
established the foundation by defining complexity in terms of program length.  
Solomonoff~\cite{solomonoff1964a,solomonoff1964b} introduced the universal prior,
assigning to each program $p$ of length $|p|$ a weight $2^{-|p|}$,
thereby embedding a probabilistic measure into computation.
Levin~\cite{levin1974} further clarified the probabilistic and conservation aspects of algorithmic information,  
connecting coding and probability in a unified framework.  

A thermodynamic analogy was made explicit by Tadaki~\cite{tadaki2002,tadaki2008,tadaki2019},  
who defined partition functions of the form  
$Z(\beta)=\sum_p e^{-\beta |p|}$ and interpreted $\beta$ as an externally imposed inverse temperature.  
Manin and Marcolli~\cite{manin2009a,manin2009b} further recast this analogy  
within the framework of $C^*$-dynamical systems, in which program lengths act as generators  
of a one-parameter automorphism group and equilibrium states are characterized  
by the Kubo--Martin--Schwinger (KMS) condition.  
Their construction provides an elegant operator-algebraic formulation of algorithmic thermodynamics,  
highlighting deep connections between computation, information, and noncommutative geometry.  
Nevertheless, the temperature in this framework is introduced formally---within the algebraic apparatus--- 
rather than emerging from the internal combinatorial structure of computation itself.  
It thus represents a mathematically refined but externally parameterized description.

Baez and Stay~\cite{baez2012} subsequently developed a complementary and more accessible formulation,  
interpreting program statistics in direct analogy with energy ensembles in statistical mechanics  
and introducing a dictionary between algorithmic and thermodynamic quantities.  
Their ``algorithmic thermodynamics'' made the analogy operational, but, as in earlier approaches,  
the temperature remained an imposed parameter rather than an emergent property.

A step toward a more structural understanding was taken with the introduction of  
\emph{regular UTMs}~\cite{imafuku2025a}, designed to clarify the redundancy structure of programs  
and to analyze algorithmic distributions in thermodynamic terms.  
The purpose of this construction was not to derive a temperature, but to provide  
a canonical way to factor programs into \emph{core codes} and \emph{regular wrapper families}.  
This framework, however, implicitly contained the structural ingredients for temperature to arise:  
the exponential growth of wrapper multiplicities.

In previous work~\cite{imafuku2025a}, a finite-temperature framework was introduced  
to describe measures of algorithmic similarity, starting from the observation  
that program distributions can be expressed in Boltzmann form.  
That approach provided new insights into how temperature can modulate algorithmic similarity,  
but the temperature itself was treated phenomenologically: it was introduced as an external parameter  
rather than derived from the internal structure of computation.  
The present study addresses this gap by showing that, once regular UTMs are  
adopted as the underlying computational model, an intrinsic parameter  
$\gamma=\ln\mu$ naturally emerges from their redundancy structure.  (Here, `adopting' a regular UTM does not mean physically executing computation on that machine, but rather adopting it as an interpretive framework that defines how program space is decomposed and observed.)
We interpret this $\gamma$ as an \emph{algorithmic temperature} determined not externally  
but by the adopted machine itself.

In contrast to previous approaches, algorithmic temperature here emerges directly from  
the machine's internal structure, in close parallel to the way physical temperature  
arises from the combinatorial growth of microstates in statistical mechanics.  
The redundancy structure of the adopted regular UTM thus plays the role of an  
\emph{algorithmic environment}, fixing the effective temperature at which the  
algorithmic world is observed.

Traditional algorithmic information theory emphasizes the universality of UTMs:  
descriptive complexity differs only up to an additive constant across machines.  
Our approach preserves this universality but reveals an additional relative structure:  
once redundancy is taken into account, different regular UTMs induce different effective temperatures.  
Thus universality is not abandoned but refined---each adopted UTM defines a distinct  
computational environment that fixes the thermodynamic laws under which programs are observed.

\paragraph{Relativity of the core--wrapper distinction.}
The decomposition of a program into a ``core'' and a ``wrapper'' is not absolute  
but depends on the adopted regular UTM.  
Each regular UTM defines which parts of a program are regarded as essential and which as redundant,  
thereby determining both the effective complexity and the emergent temperature $\gamma$.  
In this sense, a regular UTM functions as an \emph{observational framework} that partitions program space  
into meaningful (core) and inessential (wrapper) components.  
Different regular UTMs therefore correspond to distinct computational perspectives---  
analogous to different coarse-grainings in physics that yield different thermodynamic descriptions.

In this perspective, interpreting program length as an ``energy'' and introducing a temperature  
corresponds to viewing the program space through the lens of a chosen regular UTM.  
Each such adoption provides a distinct coarse-graining of program redundancy,  
thereby endowing the algorithmic world with its own notion of equilibrium and temperature.

\section{Background on Regular UTMs}

Universal Turing machines (UTMs) provide the foundation of algorithmic information theory,
but in their most general form they contain a large degree of redundancy.
A single output $o$ can be produced by infinitely many programs, often by trivially extending
a shorter program with irrelevant suffixes.
For statistical-mechanical analogies, however, such redundancy must be handled in a principled way;
otherwise, counting arguments become ill-defined.

To address this difficulty, the author's earlier work~\cite{imafuku2025a} introduced the notion of a \emph{regular UTM}.
The key idea was to separate the essential, functional part of a program from its regular, redundant
extensions.
In a regular UTM, every program $p$ admits a unique factorization $p = w \Vert q$, into a \emph{wrapper} $w$ drawn from a fixed, output-independent regular language $W$,
and a \emph{core} code $q$ executed by a fixed reference machine $U_0$, such that $U(w\Vert q) = U_0(q)$.
Here $p=w\Vert q$ should not be read as a mere syntactic concatenation.
Rather, the wrapper $w$ specifies a regular family of redundant transformations---
such as inserting NOPs, adding comments, or appending irrelevant suffixes---applied to the core program $q$.
The regularity condition ensures that this factorization is always well-defined and unique.

In the present framework, we fix the core alphabet as binary,
$\Sigma_{\mathrm{core}} = \{0,1\}$, so that $|q|$ is measured in bits,
and let the wrapper language $W \subseteq \Sigma_{\mathrm{wrap}}^{*}$ be defined
over an alphabet $\Sigma_{\mathrm{wrap}}$ of size $b \ge 2$.
Wrapper length $|w| = \Delta$ is thus measured in $\Sigma_{\mathrm{wrap}}$-symbols.
This separation emphasizes that cores and wrappers play distinct roles:
the cores carry the essential algorithmic content,
while the wrappers contribute the redundant multiplicities that define the computational environment.

For each $\Delta \ge 0$, let $a_\Delta$ denote the number of wrappers of length $\Delta$.
A key structural property of regular UTMs is that these wrapper families grow at most exponentially,
$a_\Delta \;\asymp\; c\,\mu^\Delta, \qquad 1 \le \mu < b$,
independently of the output $o$.
Consequently, the emergent parameter $\gamma = \ln \mu$ ranges over
$0 \le \gamma < \ln b$.
This growth rate will provide the structural source of the Boltzmann factor
that emerges in the next section.

The \emph{ground/core length} $K(o)$ is the minimal $|q|$ such that $U_0(q)=o$.
This decomposition isolates the ``essential'' part of a program (the core)
from the ``redundant'' degrees of freedom (the wrappers).

In analogy with statistical mechanics, the core programs $q$ may be viewed as
\emph{macro-descriptions}, while the wrappers $w$ play the role of
\emph{micro-variations}.
A single core thus corresponds to a macrostate, and the multiplicity of its
wrappers provides the microstate degeneracy that underlies the Boltzmann factor.
This structural view, first formalized in the author's earlier work,
here becomes the foundation for deriving an intrinsic notion of temperature.

\section{Emergent Algorithmic Temperature}

Before presenting the formal derivation, it is helpful to visualize
how a regular UTM transforms a uniform ensemble of programs into
a Boltzmann--weighted distribution over core lengths.
Figure~\ref{fig:factorization_lens} summarizes this process schematically.

\begin{figure}[t]
  \centering
  \includegraphics[width=0.7\linewidth]{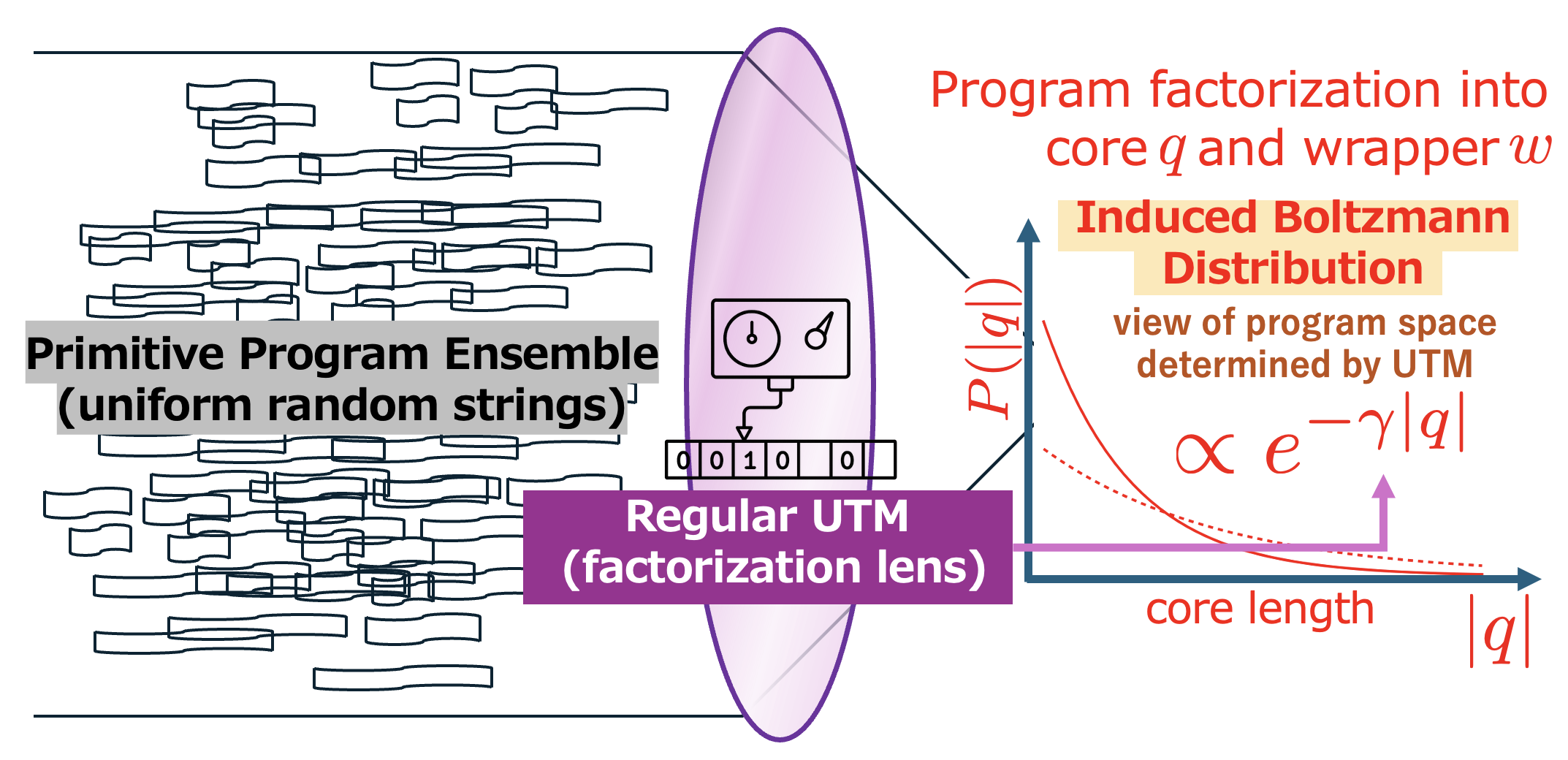}
\caption{
Schematic overview of how algorithmic temperature emerges
from a regular UTM. 
The primitive ensemble consists of all strings up to length $D$,
sampled uniformly, regardless of whether they are accepted programs.
When viewed through a regular UTM (the ``factorization lens''),
the admissible programs are uniquely decomposed into a core $q$ and a wrapper $w$
with $U(p)=U_0(q)$. 
The exponential redundancy of admissible wrappers induces
a Boltzmann--like distribution over core lengths,
$P(|q|)\propto e^{-\gamma|q|}$, where $\gamma=\ln\mu$
is determined by the growth rate of wrappers.
This transformation corresponds to observing the computational world
through a particular algorithmic environment, fixed by the chosen UTM.}
  \label{fig:factorization_lens}
\end{figure}

We now show how regular UTMs naturally induce a Boltzmann--like factor for outputs.
We start from the uniform measure over programs up to a cutoff and then take the cutoff to infinity.
The exponential growth of wrappers guarantees the emergence of an effective temperature parameter.

\begin{theorem}[Emergent Algorithmic Temperature]
Let $U$ be a regular UTM whose wrapper language over an alphabet of size $b$ has counting sequence $(a_\Delta)_{\Delta\ge0}$ with
$a_\Delta \asymp c\,\mu^\Delta$ for some $1 \le \mu < b$, and set $\gamma:=\ln \mu$.
Consider the uniform distribution over all programs $p=w\Vert q$ of the form ``wrapper $w$'' followed by ``core $q$'' with \emph{total excess length} at most $D$, i.e. $|w|+|q|\le D$.
Then the induced distribution over outputs $o$ converges, as $D\to\infty$, to
\[
P(o) \;\propto\; \sum_{L \ge 0} m_o^{(L)} \, e^{-\gamma L},
\]
where $m_o^{(L)}$ is the number of core programs of length $L$ that produce $o$ on the reference machine $U_0$.
\end{theorem}

\begin{proof}[Proof sketch]
We proceed in three steps.

\smallskip
\noindent
1. \emph{Uniform measure with cutoff.}
For each cutoff $D$, count all programs $p=w\Vert q$ with wrapper length $|w|=\Delta$ and core length $|q|=L$ such that $\Delta+L\le D$.
The number of such programs producing output $o$ is
\[
  N_o(D) = \sum_{L+\Delta\le D} m_o^{(L)}\,a_\Delta
         = \sum_{L=0}^{D} m_o^{(L)}\,A_{D-L},
\]
where $A_n := \sum_{\Delta=0}^{n} a_\Delta$ is the cumulative count of wrappers up to length $n$.

\smallskip
\noindent
2. \emph{Wrapper growth and tail asymptotics.}
By regularity of the wrapper language (See Appendix A.9 of \cite{imafuku2025a}  for proof of exponential growth), there exist constants $\mu\in[1,b)$ and $C_\pm>0$ such that for all sufficiently large $\Delta$,
\[
  C_-\,\mu^\Delta \;\le\; a_\Delta \;\le\; C_+\,\mu^\Delta.
\]
Equivalently, $a_\Delta = c\,\mu^{\Delta}(1+\varepsilon_\Delta)$ with $\varepsilon_\Delta\!\to\!0$ as $\Delta\!\to\!\infty$. 
Summing gives
\[
  A_n = \Theta\!\bigl(\Phi(n)\bigr), \qquad
  \Phi(n) :=
  \begin{cases}
    \dfrac{\mu^{\,n+1}}{\mu-1}, & \mu>1,\\[6pt]
    n+1, & \mu=1.
  \end{cases}
\]

\smallskip
\noindent
3. \emph{Cancellation and emergence of the Boltzmann factor.}
\begin{itemize}
\item
\underline{Case $\boldsymbol{\mu>1}$.}
By the regularity assumption, $A_{D-L}=\tfrac{\mu^{D-L+1}}{\mu-1}(1+o(1))$
uniformly for $L$ in any finite range.
The core programs form a prefix-free set, which ensures that their total number
does not exceed exponential order.
For each output $o$, we denote by $m_o^{(L)}$ the number of cores of length $L$
that produce it.
We do not require exponential bounds on $m_o^{(L)}$ nor the convergence of each $\sum_L m_o^{(L)}\mu^{-L}$ separately. It suffices that the Abel ratios exist: letting $M_o(z)=\sum_{L\ge0} m_o^{(L)}z^L$, the limit $\lim_{r\uparrow \mu^{-1}} M_o(r)/\sum_{o'}M_{o'}(r)$ exists. Under the wrapper asymptotics $A_n=\frac{\mu^{n+1}}{\mu-1}(1+o(1))$, this Abel ratio equals the cutoff limit of $P_D(o)$; hence the Boltzmann weight follows even when the partition sum diverges.
Under this mild condition, the total relative error in $N_o(D)$
vanishes as $D\to\infty$.
Substituting the asymptotic form then yields
\[
N_o(D)
= \Theta(\mu^{D})\sum_{L=0}^{D} m_o^{(L)}\,\mu^{-L}(1+o(1)).
\]
The common $\Theta(\mu^D)$ factor is independent of $o$ and cancels in
$P_D(o):=N_o(D)/\sum_{o'}N_{o'}(D)$, giving
\[
\lim_{D\to\infty}P_D(o)
\;=\;\frac{\sum_{L\ge0} m_o^{(L)}\,\mu^{-L}}
{\sum_{o'}\sum_{L\ge0} m_{o'}^{(L)}\,\mu^{-L}}
\;=\;\frac{\sum_{L\ge0} m_o^{(L)}\,e^{-\gamma L}}
{\sum_{o'}\sum_{L\ge0} m_{o'}^{(L)}\,e^{-\gamma L}},
\quad \gamma=\ln\mu.
\]
\item
\underline{Case $\boldsymbol{\mu=1}$.}
Here $A_{D-L}\sim D-L+1$, so
\[
N_o(D)
= \sum_{L=0}^{D} m_o^{(L)}(D-L+1)
= (D+1)\sum_{L=0}^{D} m_o^{(L)} - \sum_{L=0}^{D} L\,m_o^{(L)}.
\]
For large D, the leading factor (D+1) is common to all outputs
and thus cancels when ratios are taken.
Provided that the relative proportions of $m_o^{(L)}$ remain well-defined
(even if their total sum diverges), the limiting distribution takes the same structural form,
\[
\lim_{D\to\infty} P_D(o)
= \frac{\sum_{L\ge0} m_o^{(L)}}
{\sum_{o'}\sum_{L\ge0} m_{o'}^{(L)}}
= \frac{\sum_{L\ge0} m_o^{(L)}\,e^{-\gamma L}}
{\sum_{o'}\sum_{L\ge0} m_{o'}^{(L)}\,e^{-\gamma L}},
\quad \gamma=0.
\]
This limiting case corresponds to a zero inverse temperature,
where wrapper redundancy grows only linearly and no exponential suppression occurs.
The normalization may diverge, but the formal Boltzmann structure of the distribution remains intact.
\end{itemize}
\end{proof}

\begin{figure}[t]
  \centering
  \includegraphics[width=0.7\linewidth]{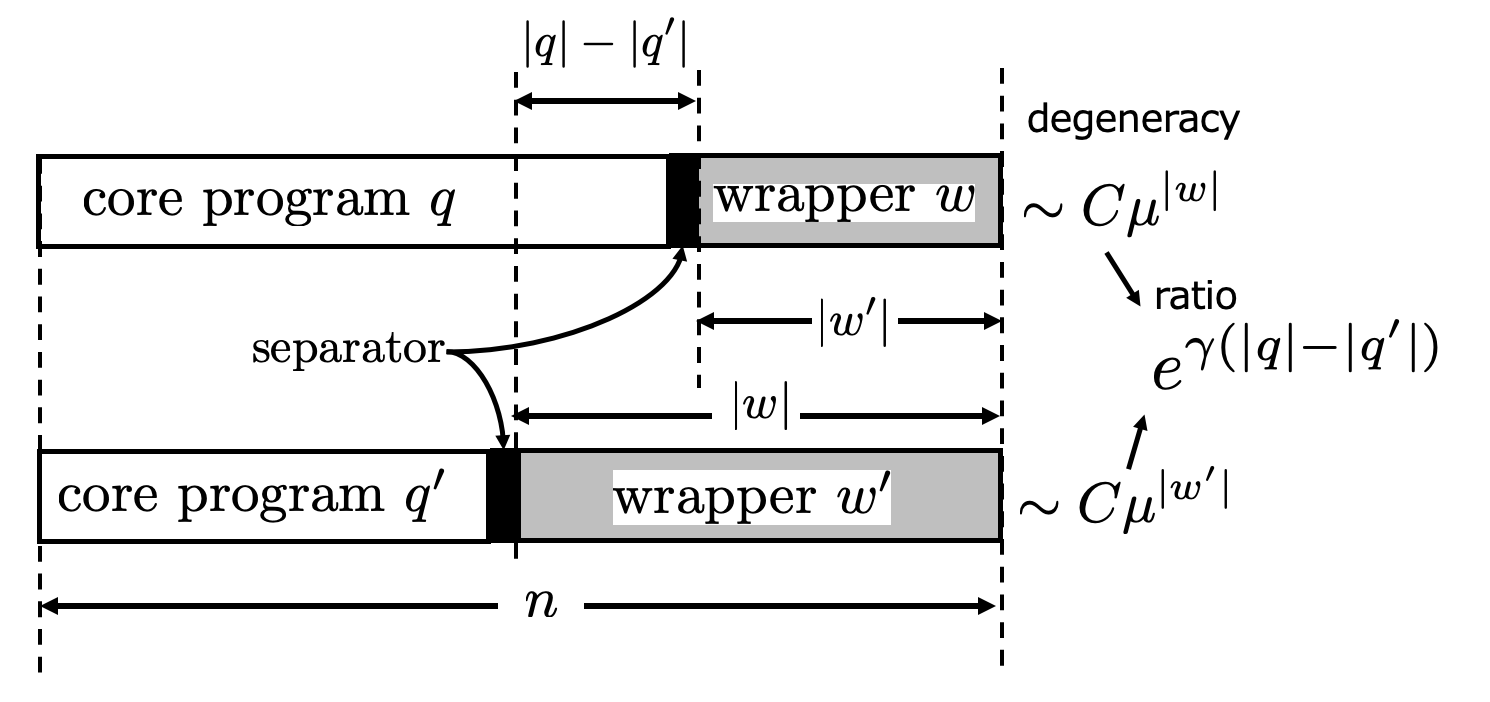}
  \caption{
  Schematic illustration of how wrapper redundancy induces
  Boltzmann weighting. Two programs with cores $q$ and $q'$ 
  are extended by wrappers $w$ and $w'$, respectively, 
  so that their total lengths are equal.
  The number of admissible wrappers of length $|w|$ 
  grows as $\sim C\mu^{|w|}$.
  Consequently, the degeneracy ratio between cores of different lengths
  scales as $e^{\gamma(|q|-|q'|)}$ with $\gamma=\ln\mu$,
  yielding the Boltzmann factor $e^{-\gamma |q|}$ in the induced distribution.}
  \label{fig:degeneracy_ratio}
\end{figure}

Intuitively, for two core programs $q$ and $q'$ of different lengths,
their possible wrappers must compensate for the difference
so that the total program length remains within the cutoff.
Since the number of admissible wrappers of length $\Delta$
grows exponentially as $\mu^{\Delta}$, 
the relative abundance of $q$ versus $q'$ scales as 
$e^{\gamma(|q|-|q'|)}$ (Fig.~\ref{fig:degeneracy_ratio}).
This directly leads to the Boltzmann factor $e^{-\gamma |q|}$.

Note that the normalization factor
\[
Z(\gamma) := \sum_{o}\sum_{L\ge0} m_o^{(L)} e^{-\gamma L}
\]
plays exactly the role of a partition function. 
In this analogy, the ``energy'' is identified with the core length $L$, 
and the exponential weight $e^{-\gamma L}$ is the Boltzmann factor. 
Hence the parameter $\gamma$ is naturally interpreted as an 
\emph{inverse algorithmic temperature}. 
Thus $Z(\gamma)$ serves as a genuine partition function, 
normalizing the global distribution over outputs. 
Equivalently, one may define
\[
Z_o(\gamma) := \sum_{L\ge0} m_o^{(L)} e^{-\gamma L},
\]
so that $P(o)=Z_o(\gamma)/Z(\gamma)$, making explicit the connection to \cite{imafuku2025a}. 

As a special case, if the wrapper language is the full binary language, then $a_\Delta=2^\Delta$ and hence $\mu=2$, 
so that $\gamma=\ln 2$. 
In this case the induced weight for core programs of length $L$ is exactly $2^{-L}$, 
yielding Solomonoff's universal distribution:
\[
P(o) \;\propto\; \sum_{L\ge0} m_o^{(L)} 2^{-L}.
\]
This shows that Solomonoff’s prior arises as the minimum-temperature limit for the case $b=2$, corresponding to the full binary wrapper language.
For larger wrapper alphabets ($b > 2$), even lower temperatures ($\gamma > \ln 2$) are in principle attainable.

In analogy with physics, the wrapper degrees of freedom can be viewed as ``dressing'' the core program. 
Eliminating these degrees of freedom yields an effective distribution with a Boltzmann factor. 
In this sense, our construction can be viewed as a renormalization procedure in the algorithmic domain.

\section{Algorithmic Thermodynamics: Interpretation}

Theorem~1 establishes that every regular UTM induces an inverse temperature 
$\gamma=\ln\mu$, determined entirely by the exponential growth rate of wrappers. 
This section interprets this result in thermodynamic and information-theoretic terms.

\subsection{Algorithmic equilibrium}

The limiting distribution
\[
P(o) \;\propto\; \sum_{L\ge0} m_o^{(L)} e^{-\gamma L}
\]
takes exactly the Boltzmann form, with program length playing the role of 
energy. This justifies speaking of an \emph{algorithmic equilibrium}: 
given a uniform ensemble of programs, the induced statistics of outputs 
stabilize to a distribution governed by the machine's intrinsic temperature~$\gamma$.
In this sense, a regular UTM defines an effective thermodynamic bath 
that determines the equilibrium structure of program space.

\subsection{Temperature as a machine property}

Unlike previous approaches in which temperature was externally imposed, 
here $\gamma$ depends only on the structure of the adopted regular UTM. 
It is independent of the output $o$ and thus represents a universal 
thermodynamic parameter attached to the computational model itself. 
This parallels physical thermodynamics, where temperature is a property 
of the bath rather than of individual microstates. 

From this standpoint, a regular UTM acts as an \emph{algorithmic environment}: 
its redundancy structure defines the degrees of freedom regarded as 
thermodynamically irrelevant. 
The exponential growth of admissible wrappers thus determines 
the temperature at which program space is observed.

\subsection{Relation to Solomonoff's universal prior}

Solomonoff's universal prior assigns each program $p$ a weight $2^{-|p|}$,
justified by Kraft's inequality for prefix-free codes:
shorter programs carry higher probability under a uniform distribution 
on prefix-free bit strings.
In that framework, the exponential decay with program length is imposed 
as a coding-theoretic constraint.

By contrast, in our framework the factor $e^{-\gamma L}$ is not imposed
through Kraft's inequality but arises intrinsically from the redundancy
of the wrapper language.
The core programs are assumed to form a prefix-free set,
ensuring that their total number does not grow faster than exponentially.
We only require that the weighted sums
$\sum_L m_o^{(L)}e^{-\gamma L}$ be well defined---convergent or, if divergent,
regularizable in a consistent manner---so that relative probabilities 
can be meaningfully compared.
The exponential weight itself, however, is determined solely
by the structural redundancy of the regular UTM, not by any coding convention.

\paragraph{Remark (Solomonoff as a limit).}
Our construction includes Solomonoff's prior as a limiting case: 
if the wrapper language is the full binary language, $a_\Delta=2^\Delta$, 
then $\mu=2$ and $\gamma=\ln 2$, yielding Solomonoff's measure.
Hence the universal prior corresponds to the case of maximal redundancy 
(or minimum temperature), while regular UTMs with $1\le\mu<2$ represent 
higher-temperature systems. 
This reveals that Solomonoff's prior is not unique, but one extremal 
point in a continuum of algorithmic thermodynamic regimes.
\subsection{Energy and redundancy}

A key conceptual choice in this construction is that the effective energy
entering the Boltzmann factor is associated with the core length $|q|$,
rather than the total program length $|p|=|w|+|q|$. 
The wrapper $w$ represents environmental degrees of freedom---redundant 
transformations such as padding, comments, or regular syntactic variations---
that do not change the semantic content of the computation. 
These correspond to the external combinatorial structure that gives rise 
to entropy and hence to temperature.

By contrast, the core $q$ encodes the essential algorithmic information
that determines the output.
Using $|p|$ as the energy variable would double-count redundancy, 
since the multiplicity of wrappers $a_\Delta$ is precisely what generates 
the Boltzmann factor $e^{-\gamma |q|}$.
Thus the correct correspondence is
\[
\text{Energy} \;\leftrightarrow\; |q|, 
\qquad
\text{Entropy source (bath)} \;\leftrightarrow\;
\text{wrapper multiplicity } a_\Delta .
\]
Eliminating wrappers corresponds to tracing over environmental
redundancies, yielding an effective equilibrium distribution for the core 
with inverse temperature $\gamma=\ln\mu$. 

\section{Broader Implications and Connections}

\subsection{Regular UTMs as algorithmic environments}

The choice of a regular UTM specifies which forms of redundancy 
are admissible and therefore determines the computational ``environment''
in which algorithms are observed. 
Different regular UTMs correspond to distinct observational frameworks, 
each defining what is treated as essential (core) and what as redundant (wrapper). 
This observer-dependence parallels the system--environment decomposition in physics: 
it is not absolute but relative to the chosen descriptive scale.

\subsection{Algorithmic temperature as epistemic resolution}

The parameter $\gamma$ should not be interpreted as a measure of 
computational power, but as a measure of \emph{epistemic resolution}: 
how finely an observer distinguishes among distinct core programs.
A high-temperature observer adopts a coarse-grained perspective, treating functionally similar
programs as equivalent and focusing on general functional similarity.
A low-temperature observer resolves finer distinctions, assigning weight 
primarily to compact cores. 
These perspectives are complementary rather than hierarchical:
high temperature favors flexibility and abstraction, 
while low temperature favors precision and parsimony.

\subsection{Relation to previous work~\cite{imafuku2025a}}

In earlier work~\cite{imafuku2025a}, finite-temperature parameters were 
introduced phenomenologically to interpolate between minimal and distributed 
descriptions of algorithmic similarity.
The present formulation clarifies the structural origin of such parameters:
once regular UTMs are adopted as the underlying model, the inverse temperature 
$\gamma$ emerges intrinsically from the redundancy structure of the machine.
Thus, the finite-temperature measures of the earlier work can now be 
interpreted as reflecting intrinsic properties of the chosen computational 
environment rather than externally imposed biases.

\medskip
\noindent
A further connection can be made by reconsidering the treatment of wrapper degrees of freedom.
In the present analysis, wrappers are fully integrated out, yielding an equilibrium
distribution over core programs alone.
By contrast, the earlier framework effectively retained wrappers as part of the ensemble,
introducing temperature directly at the level of full programs $p = w \Vert q$.
As discussed in Appendix~\ref{appendix a}, this ``wrapper-inclusive'' viewpoint corresponds to
stopping the coarse-graining process at an intermediate stage,
before tracing out environmental redundancies.
In this way, the distributions of~\cite{imafuku2025a} appear as special cases within
the general dressing hierarchy of regular UTMs, while the present work
completes that renormalization to reveal the intrinsic origin of the
algorithmic temperature~$\gamma=\ln\mu$.

\subsection{Observer-dependence and universality}

Both the effective energy variable (core length $|q|$) and the resulting 
temperature $\gamma$ are observer-dependent quantities, defined relative 
to the adopted regular UTM. 
Yet this does not undermine universality: what remains invariant is the 
\emph{form} of the thermodynamic relation between redundancy and equilibrium.
Different observers, each equipped with their own UTM, perceive distinct 
algorithmic landscapes that are internally consistent but not privileged. 
Universality thus manifests through equivalence classes of descriptions---
a structural relativity of observation rather than a breakdown of objectivity.

\section{Conclusion}

We have shown that regular UTMs naturally induce an inverse algorithmic temperature 
$\gamma=\ln\mu$, arising intrinsically from the exponential redundancy 
of wrapper families. 
Starting from a uniform ensemble of programs with an absolute cutoff,
the elimination of redundant wrappers yields Boltzmann weights over core lengths. 
In this way, temperature is not externally imposed but emerges from the 
machine's structural properties.

This framework unifies and extends earlier perspectives in algorithmic thermodynamics:
Solomonoff's universal prior appears as the limiting case of maximal wrapper growth,
while the finite-temperature measures previously introduced by the present author 
acquire a principled structural foundation.
What was formerly treated phenomenologically now arises as a consequence of 
the redundancy structure of computation itself.

More broadly, the findings suggest that notions such as equilibrium and temperature 
are not merely analogies but formally realizable in the algorithmic domain.
The emergent temperature reflects the observer's computational environment: 
the choice of regular UTM fixes the admissible forms of redundancy and thereby 
the effective temperature at which the algorithmic world is seen.

Formally, the resulting ensemble
\[
P(p)\propto e^{-\gamma |p|}
\]
shares the same exponential structure as the Boltzmann distribution
$e^{-\beta E}$ in statistical mechanics and the Euclideanized Feynman kernel $e^{-S_E/\hbar}$ in quantum theory.
In each case, the exponent represents a cost function---energy, action, or
descriptional length---while the prefactor ($\beta$, $\hbar^{-1}$, or $\gamma$)
sets the observer's effective resolution scale.
This structural analogy suggests that the algorithmic parameter $\gamma$
plays a role formally parallel to an inverse Planck constant,
linking thermodynamic and informational perspectives within a unified
exponential framework (see Appendix~\ref{appendix b}).

This observer-dependent yet structurally determined viewpoint opens 
paths for future research. 
In learning theory, finite algorithmic temperatures may provide a principled 
way to model controlled bias or generalization noise. 
In complexity theory, $\gamma$ offers a thermodynamic lens on the trade-off 
between program length and computational resources, suggesting an algorithmic 
analogue of energy--entropy balance in physics.

\begin{center}
\emph{Algorithmic temperature is not a metaphor borrowed from physics, 
but a structural reality of computation itself.}
\end{center}
\section*{Acknowledgments}
I am grateful to my colleagues at AIST for their various forms of support.
This study benefited from the use of ChatGPT-5 (OpenAI) for improving the manuscript's structure and English language editing.

\appendix
\section*{Appendix}
\section{Wrapper-inclusive distributions\label{appendix a}}

In the main text, we considered the distribution obtained after
eliminating wrapper degrees of freedom, which induces a Boltzmann weight
$e^{-\gamma L}$ over core lengths.
One may also consider an alternative viewpoint in which wrappers are
\emph{not} eliminated, but retained as part of the effective ensemble.
Such a construction naturally yields a ``wrapper-inclusive'' equilibrium
distribution.

This situation corresponds to that analyzed in \cite{imafuku2025a}, where
finite-temperature parameters were introduced directly into
distributions defined over full programs $p=w\Vert q$.
In the present framework, this can be interpreted as stopping the
coarse-graining (or renormalization) at an intermediate stage, without
integrating over the wrapper multiplicities.
The resulting ensemble thus describes the joint statistics of cores and
wrappers rather than the reduced equilibrium over cores alone.

From this perspective, the wrapper-inclusive distributions of
 \cite{imafuku2025a} are not external constructions but special cases of the
general dressing framework provided by regular UTMs.
The main difference is that the present work completes the
coarse-graining by integrating out wrappers entirely, thereby revealing
the intrinsic origin of the Boltzmann factor and the emergent algorithmic
temperature $\gamma=\ln\mu$.

\section{Toward an Algorithmic Interpretation of $\hbar$\label{appendix b}}
\noindent
While the present work focuses on equilibrium distributions in program space,
the same exponential structure suggests a deeper analogy between
thermodynamic, quantum, and algorithmic ensembles.
In each case, an inverse scale parameter --- $\beta$ in statistical mechanics,
$\hbar^{-1}$ in quantum theory, and $\gamma$ in algorithmic thermodynamics ---
governs the resolution at which configurations are distinguished.
From this viewpoint, $\hbar$ and $\gamma$ play parallel structural roles:
both quantify how finely the observer resolves the underlying configuration space.

\medskip
\noindent
The algorithmic ensemble derived in this work,
\[
P(p) \propto e^{-\gamma |p|},
\]
shares the same exponential form as the Boltzmann distribution
$e^{-\beta E}$ in statistical mechanics
and the Euclideanized Feynman kernel $e^{-S_E/\hbar}$ in quantum theory~\cite{feynman1965}.
Each represents a measure over possible realizations --- microstates, paths,
or programs --- governed by an inverse scale parameter that reflects
the observer's descriptive granularity.

\begin{center}
\renewcommand{\arraystretch}{1.0}
\setlength{\tabcolsep}{4pt}
\small
\begin{tabular}{@{}lccc@{}}
\toprule
 & \textbf{Statistical Mechanics} & \textbf{Quantum Theory} & \textbf{Algorithmic Thermodynamics} \\
\midrule
Kernel & $e^{-\beta E}$ & $e^{-S_E/\hbar}$ & $e^{-\gamma |p|}$ \\
Quantity & Energy $E$ & Euclidean Action $S_E$ & Program Length $|p|$ \\
Scale Parameter & $\beta = 1/k_BT$ & $\hbar^{-1}$ & $\gamma$ \\
Domain & Microstates & Paths $x(\tau)$ & Programs $p$ \\
Interpretation & Thermal ensemble & Euclidean path integral & Algorithmic ensemble \\
\bottomrule
\end{tabular}
\normalsize
\end{center}

\noindent
This correspondence can be summarized schematically as
\[
e^{-\beta E}
\;\;\longleftrightarrow\;\;
e^{-S_E/\hbar}
\;\;\longleftrightarrow\;\;
e^{-\gamma |p|}.
\]
In all three cases, the exponent represents a cost functional ---
energy, action, or descriptional length --- and the prefactor
($\beta$, $\hbar^{-1}$, or $\gamma$) determines the effective
resolution of observation.

\medskip
\noindent
From this structural viewpoint, $\hbar$ may be interpreted as a measure
of the observer's descriptive precision in the quantum domain,
just as $\gamma$ characterizes it in the algorithmic one.
Both constants are thus not merely physical or computational parameters
but structural indicators of how finely the world is resolved.
In the algorithmic case, $\gamma$ emerges intrinsically from the redundancy
structure of the chosen regular UTM, suggesting that an analogous
informational origin of $\hbar$ might underlie physical quantization itself.
This structural correspondence echoes the stochastic quantization view of field theory, where the Euclidean path integral appears as the stationary distribution of a diffusive process~\cite{parisi1981,namiki1992}.
\bibliographystyle{unsrtnat} 
\bibliography{refs202510}

\end{document}